\documentclass{llncs}
\usepackage{amsfonts} 
\usepackage{mathtools}
\usepackage{complexity}
\usepackage{mathrsfs}  
\usepackage[numbers,sectionbib]{natbib}

\usepackage[backgroundcolor=white,linecolor=black]{todonotes}
\usepackage{xspace}

\usepackage{tikz}
\usetikzlibrary{shapes.geometric, calc, fit}
\usetikzlibrary{shadows}
\pgfdeclarelayer{bg}    
\pgfsetlayers{bg,main}

\tikzstyle{regular vertex}=[circle,
draw]
\tikzstyle{black vertex}=[regular vertex,
fill=black]
\tikzstyle{white vertex}=[regular vertex,
fill=black!1]
\tikzstyle{vertex}=[regular vertex,
minimum size=12pt,
inner sep=0pt,
circular drop shadow]
\tikzstyle{svertex}=[vertex,
fill=black!23]
\tikzstyle{bvertex}=[vertex,
fill=black!3]
\tikzstyle{edge}=[draw,thick]
\tikzstyle{selected edge}=[draw, line width=3pt]

\DeclareMathOperator{\discounts}{\mathscr{D}}
\DeclareMathOperator{\pc}{pc}
\DeclareMathOperator{\OPT}{opt}

\newcommand{\pblong}{{\sc Clever Shopper}\xspace}
\newcommand{\pbabbrev}{{\sc Clever Shopper}\xspace}

\usepackage{enumitem}

\newlist{aims}{itemize}{1}
\setlist[aims,1]{
	label={Question~\arabic*}
	,leftmargin=*
	,align=left
	,   labelsep=0mm
}

\newcommand{\decisionproblem}[3]{{
\centering
		\noindent
        \parbox{35em}{
				\begin{minipage}[t]{1\linewidth}
					{\sc #1}
					\vspace{1pt}
					\begin{aims}[topsep=0pt]
						\item[Input:] {#2}
						\item[Question:] {#3}
					\end{aims}
				\end{minipage}
			}
}}

\usepackage{cleveref}  

\begin{document}
	
	\title{The Clever Shopper Problem} 
\author{Laurent Bulteau\inst{1}\and Danny Hermelin\inst{2} \and Anthony Labarre\inst{1} \and St\'ephane Vialette\inst{1}}
\institute{
	Universit\'e Paris-Est, LIGM (UMR 8049), CNRS, ENPC, ESIEE Paris, UPEM, F-77454, Marne-la-Vall\'ee, France\\
	\texttt{\{laurent.bulteau,anthony.labarre,stephane.vialette\}@u-pem.fr}
\and
Department of Industrial Engineering and Management, Ben-Gurion University of the Negev, Israel\\
	\texttt{hermelin@bgu.ac.il}}
\authorrunning{L. Bulteau et. al.} 
	
	\maketitle              

	\begin{abstract}
	We investigate a variant of the so-called {\sc Internet Shopping}
    problem introduced 
    by Blazewicz et al. (2010), 
	where a customer wants to buy a list of products at the lowest 
    possible total cost from shops which offer discounts when 
    purchases exceed a certain threshold. 
	Although the problem is \NP-hard, we provide exact algorithms for
    several cases, e.g. when each shop sells only two items, 
    and an \FPT\ algorithm for the number of items, or for the number of 
    shops when all prices are equal. We complement each result with
    hardness proofs in order to draw a tight boundary between tractable 
    and intractable cases. Finally, we give an approximation 
    algorithm and hardness results for the problem 
    of maximising the  sum of discounts. 
	\end{abstract}

    \section{Introduction}
    
\citet{DBLP:journals/amcs/BlazewiczKMUW10} introduced and  described     the {\sc Internet Shopping} problem as follows: given a set of shops offering products at various prices and the delivery costs for each set of items bought from each shop, find where to buy each product from a shopping list at a minimum total cost. 
    The problem is known to be 
    \NP-hard in the strong sense even 
    when all products are free and all delivery costs are equal to one, and
	admits no polynomial $(c\ln n)$-approximation algorithm (for any $0<c<1$) unless \P\ = \NP. 
    
    A more realistic variant 
    takes into account 
    discounts offered by shops in some cases. These could be offered, for instance, when the shopper's purchases exceed a certain amount, or in the case of special promotions where buying several items together costs less than buying them separately.  \citet{DBLP:journals/4or/BlazewiczBKM14} investigated such a variant, 
    which features a concave increasing discount function on the products' prices. They showed that the 
    problem is  \NP-complete in the strong sense even if each product appears in at most three shops and each shop sells exactly three products, as well as in the case where each product is available at three different prices and each shop has all products but sells exactly three of them at the same value. 
	A variant where two separate discount functions are taken into account (one for the deliveries, the other for the prices) was also recently introduced and studied by \citet{Blazewicz2016}.

In this work, we investigate the case where 
a shopper aims to buy $n$ books from $m$ shops 
with free shipping; additionally, each shop offers a discount when purchases exceed a certain threshold (discounts and thresholds are specific to each shop).  We show that the associated decision problem, which we call the \pbabbrev problem, is already \NP-complete when 
only two shops are available, 
or when all books are available from two shops and each shop sells exactly three books. 
We also obtain parameterised hardness results: namely, that \pbabbrev is \W[1]-hard when the parameter is $m$ or the number of shops in a solution, and that it admits no polynomial-size kernel.  On the positive side, we give a polynomial-time algorithm for the case where every shop sells at most two books, an \XP\ algorithm for the case where few shops sell books at small prices, an \FPT\ algorithm with parameter $n$, and another \FPT\ algorithm with parameter $m$.   

    Let us now formally define 
    \pbabbrev. 
    For $n\in \mathbb{N}$, let $[n]=\{1$, $2$, $\ldots$, $n\}$. 
	Let $B$ be a set of books to buy, 
    $S$ be a set of shops, $E\subseteq B\times S$ encodes the availability of the books in the shops, and $w:E\rightarrow \mathbb N$ encodes the prices. 
    Choosing a shop in which to buy each book is encoded as a subset $E'\subseteq E$, such that each book is \emph{covered exactly once} (i.e., any $b\in B$ has degree 1 in $E'$).
	A \emph{discount} $d_s\in\mathbb{R^+}$ is associated to each shop $s$ and offered
	when a \emph{threshold} $t_s\in\mathbb{R^+}$ is reached, which is formally defined using the following \emph{threshold function}:
	
	$$\delta(s, E', d_s, t_s)=\left\{
	\begin{array}{ll}
	d_s & \mbox{ if } \sum_{(b, s) \in E'} w(e)\geq t_s,\\
	0 & \mbox{ otherwise}.
	\end{array}
	\right.$$
	
	\noindent We refer to the function $\discounts$ that maps each shop $s$ to the pair $(d_s,t_s)$ as the \emph{discount function}. The problem we study is formally stated below, and  generalises well-studied problems such as {\sc bin covering}~\cite{ASSMANN1984502} and {\sc $H$-index manipulation}~\cite{VANBEVERN201619}.
    
    \medskip
    
	\decisionproblem{\pblong}{
		an edge-weighted bipartite graph $G=(B\cup S, E, w)$;
		a discount function $\discounts$; 
		a bound $K\in\mathbb{N}$.
	}{
		is there a subset $E'\subseteq E$ that covers each element of $B$ exactly once and such that $\sum_{e\in E'} w(e)-\sum_{s\in S} \delta(s, E', d_s, t_s)\leq K$?}
	
    \section{Hardness Results}
    \label{sec:hardness-results}
	
    We prove in this section several hardness results under various restrictions, both with regards to classical complexity theory and parameterised complexity theory. 
	
	We 
    show that \pbabbrev is \NP-complete even if there are only two shops to choose from. For this first hardness result we need book prices to be encoded in binary (i.e. they can be exponentially high compared to input size).
    
	\begin{proposition}\label{prop:problem-is-hard-for-two-shops}
		\pblong is \NP-complete in the weak
        sense (\emph{i.e.}, prices are encoded in binary), even when
		$|S|=2$.
      
	\end{proposition}
	\begin{proof}[reduction from {\sc Partition}]
        Recall the well-known \NP-complete {\sc Partition}
		problem~\cite{DBLP:conf/coco/Karp72}: 
		given a finite set $A$ and a size $\omega(a)\in\mathbb{N}$ for 
        each element in $A$,
        decide whether there exists
		a subset $A'\subseteq A$ such that
		$\sum_{a\in A'}\omega(a)=\sum_{a\in A\setminus A'}\omega(a)$.
		
		Let $\mathcal{I}=(A, \omega)$ be an instance of {\sc Partition}, 
		and $T=\sum_{a\in A}\omega(a)$.
		We obtain an instance $\mathcal I'$ of \pblong as
		follows: 
        introduce two shops $s_1$ and $s_2$ 
        with $(d_{s_1}, t_{s_1})=(d_{s_2}, t_{s_2})=(1, T/2)$. 
		Each item $a \in A$ is a book that shops $s_1$ and $s_2$
		sell for the same price --- namely, $\omega(a)$. 
		It is now 
        clear that
		there exists a subset $A'\subseteq A$ such that
		$\sum_{a\in A'}\omega(a)=\sum_{a\in A\setminus A'}\omega(a)$
		if and only if
        all books can be purchased for a total cost of $T-2$.
	\qed
    \end{proof}
	
This \NP-hardness result allows 
arbitrarily 
high prices (the reduction from {\sc Partition} 
requires 
prices of the 
order of $2^{|B|}$). In a more realistic setting, we might assume a polynomial bound on prices, i.e., they can be encoded in unary. 
As we show below, the problem remains hard for a few shops 
in the sense of \W[1]-hardness. We 
complement this result with an {\XP} algorithm in \Cref{prop:xp-fewshops}.

\begin{proposition}\label{prop:problem-is-strongly-hard-for-k-shops}
	{\pblong}	 is \W[1]-hard for $m=|S|$ in the strong sense 
    (\emph{i.e.}, even when prices are encoded in unary).         
\end{proposition} 
\begin{proof}[reduction from {\sc Bin Packing}]
    Recall the well-known {\sc Bin Packing} problem:     given 
    $n$ items with weights $w_1, w_2, \ldots, w_n$ and
    $m$ bins  with the same given capacity $W$, decide whether each item can be assigned to a bin so that the total weight of the items in any bin does not exceed $W$. 
    {\sc Bin Packing} is \NP-complete in the strong sense and \W[1]-hard for parameter $m$, even when $\sum_{i=1}^nw_i = mW$ and all weights are encoded in unary~\cite{DBLP:journals/jcss/JansenKMS13}.
    
    We build an instance $\mathcal{I}$ of \pblong from an instance of {\sc Bin Packing} with the aforementioned restrictions as follows. 
    Create $m$ identical shops, each  with $t_s=W$ and $d_s=1$. 
    Create $n$ books, where book $i$ is available in every shop at price $w_i$. 
    The budget is $m(W-1)$. In other words, any solution requires to obtain the discount from every shop, which is only possible if purchases amount to a total of exactly $W$ per shop before discount.
    Therefore, the solutions to $\mathcal{I}$ 
    correspond exactly to the solutions of the original instance of {\sc Bin Packing}.
\qed
\end{proof}

We can obtain another hardness result under the assumption that all books are sold at a unit price. Here we cannot bound the total number of shops (we give an {\FPT} algorithm for parameter $m$ in \Cref{prop:fpt-fewshops} in this setting), but only the number of \emph{chosen} shops (i.e., shops where at least one book is purchased). 

\begin{proposition}\label{prop:problem-is-param-hard-for-k-selected-shops}
	\pbabbrev with unit prices is \W[1]-hard for the parameter 
    ``\emph{number of chosen shops}''.       
\end{proposition}         

\begin{proof}[reduction from {\sc Perfect Code}]
    Given a graph $G = (V, E)$ and a positive integer $k$, 
    {\sc Perfect Code} asks for a size-$k$ subset $V' \subseteq V$ such
    that for each vertex $u \in V$ there is precisely one vertex in 
    $N[v] \cap V'$ (where $N[v]$ is the \emph{closed} neighbourhood of $v$, i.e., $v$ and its adjacent vertices, as opposed to the \emph{open} neighbourhood $N(v)=N[v]\setminus \{v\}$). 
    This problem is known to be \W[1]-hard for parameter~$k$~\cite{DBLP:journals/ipl/Cesati02}.
    
    Let $\mathcal{I} = (G = (V,E), k)$ be an instance of {\sc Perfect Code}.
    Write $V = \{u_1, u_2, \ldots, u_n\}$.
    We obtain an instance $\mathcal{I}'$ of \pbabbrev as follows.
    Let us first define a bipartite graph $G' = (B \cup S, E')$ where
    $B = \{b_i : u_i \in V\}$, $S = \{s_i : u_i \in V\}$ and
    $E' = \{\{b_j,s_i\} : u_j \in N_G[u_i]\}$.
	All shops sell books at a unit price.
    As for the discount function, for each shop $s_i \in S$ we have 
    $\mathscr{D}(s_i) = (d_G(u_i)+1, 1)$ (\emph{i.e.,} a unit discount will 
    be applied
    from $d_G(u_i) + 1$ of purchase).     
    \Cref{fig:perfect code} illustrates the construction.

\begin{figure}[t]
        \centering

  \begin{tikzpicture}
    [
        scale=.5, 
                shorten >=2pt,
                shorten <=2pt,
                node distance=.5in
    ]
    \begin{scope}
        \node [black vertex] [label=left:$u_1$] (u1) at (0,0) {};
        \node [white vertex] [label=above:$u_2$] (u2) at (1.1,1.1) {};
        \node [white vertex] [label=below:$u_3$] (u3) at (1.1,-1.1) {};
        \node [white vertex] [label=above:$u_4$] (u4) at (2.2,0) {};
        \node [black vertex] [label=right:$u_5$] (u5) at (3.3,0) {};

        \path [edge] (u1) to (u2);
        \path [edge] (u1) to (u3);
        \path [edge] (u2) to (u3);
        \path [edge] (u2) to (u4);
        \path [edge] (u3) to (u4);
        \path [edge] (u4) to (u5);
    \end{scope}

    \begin{scope}[xshift=9.5cm]
        \node[bvertex] [label=left:$b_1$] (b1) at (0,-1)
                {\includegraphics[width=.3cm]{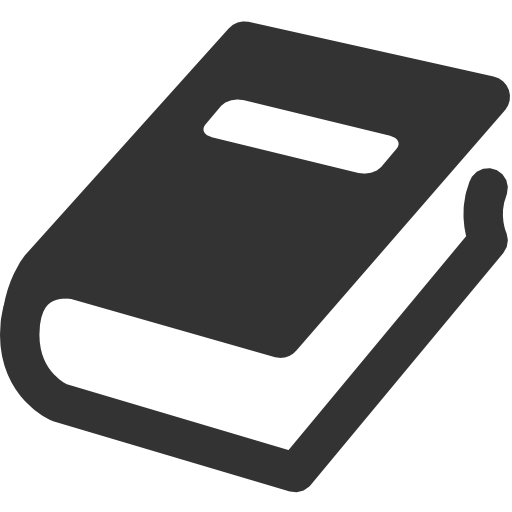}};
        \node[svertex] [label=left:$s_1$] (s1) at (0,1)
                {\includegraphics[width=.3cm]{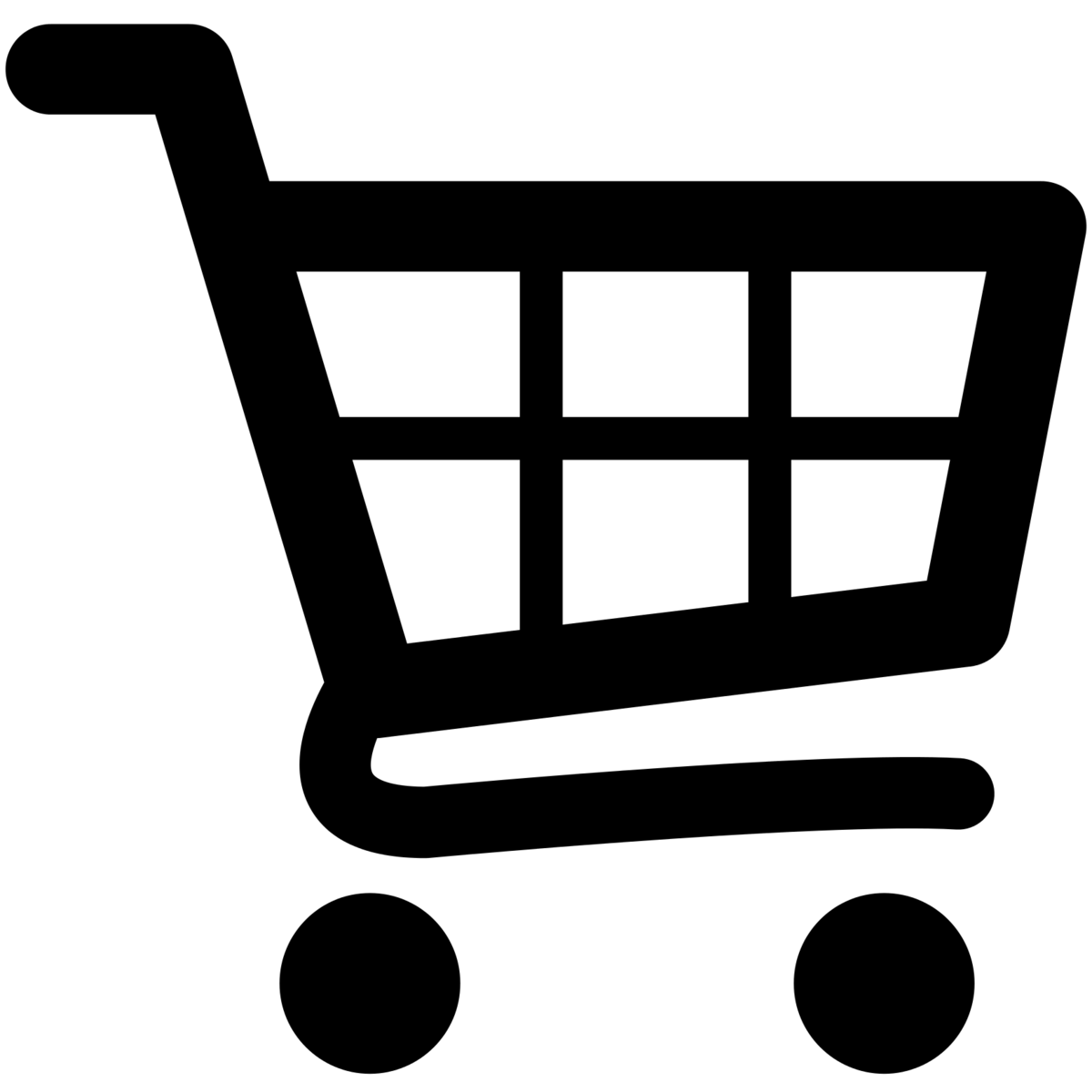}};
 
        \node[bvertex] [label=above:$b_2$] (b2) at (2,3)
                {\includegraphics[width=.3cm]{book}};
        \node[svertex] [label=above:$s_2$] (s2) at (4,3)
                {\includegraphics[width=.3cm]{basket}};
            
        \node[svertex] [label=below:$s_3$] (s3) at (2,-3)
                {\includegraphics[width=.3cm]{basket}};
        \node[bvertex] [label=below:$b_3$] (b3) at (4,-3)
                {\includegraphics[width=.3cm]{book}};
            
        \node[svertex] [label=left:$s_4$] (s4) at (6,-1)
                {\includegraphics[width=.3cm]{basket}};
        \node[bvertex] [label=left:$b_4$] (b4) at (6,1)
                {\includegraphics[width=.3cm]{book}};
            
        \node[bvertex] [label=right:$b_5$] (b5) at (8,-1)
                {\includegraphics[width=.3cm]{book}};
        \node[svertex] [label=right:$s_5$] (s5) at (8,1)
                {\includegraphics[width=.3cm]{basket}};
   
                \path [selected edge] (s1) to node [left] {$1$} (b1);
        \path [edge] (b2) to node [above] {$1$} (s2);
        \path [edge] (b3) to node [below] {$1$} (s3);  
        \path [edge] (s4) to node [left] {$1$} (b4);
        \path [selected edge] (s5) to node [left] {$1$} (b5);
                \path [edge] (b1) to node [left,pos=0.4] {$1$} (s2);        
                \path [selected edge] (s1) to node [above left] {$1$} (b2); 
        \path [edge] (b1) to node [below left,pos=0.4] {$1$} (s3);        
                \path [selected edge] (s1) to node [left,pos=0.4] {$1$} (b3); 
        \path [edge] (b2) to node [right] {$1$} (s3);        
                \path [edge] (s2) to node [left] {$1$} (b3); 
                \path [edge] (b2) to node [left,pos=0.4] {$1$} (s4);        
                \path [edge] (s2) to node [above right] {$1$} (b4); 
                \path [edge] (b3) to node [below right] {$1$} (s4);        
                \path [edge] (s3) to node [left,pos=0.4] {$1$} (b4);         
                \path [selected edge] (b4) to node [above] {$1$} (s5);        
                \path [edge] (s4) to node [below] {$1$} (b5); 
        \end{scope}

  \end{tikzpicture}
  \caption{\label{fig:perfect code}%
        Reducing {\sc Perfect Code} to \pbabbrev.
    Left: The input graph with a size-$2$ perfect code (bold).
    Right: The corresponding bipartite graph and  a solution
    with total cost $5 - 2 = 3$ (bold).
  }
\end{figure}

    We claim that there exists a size-$k$ perfect code for $G$ if and 
    only if all books can be bought for a total cost of $n-k$.

        \noindent\fbox{$\Rightarrow$}
    Let $V' \subseteq V$ be a size-$k$ perfect code in $G$.
    For every $u_i \in V$, let $u_{\pc(i)}$ be the unique vertex in 
    $N[v] \cap V'$ ($\pc$ is well-defined since $V'$ is a perfect code).
    Then buying each book $b_i \in B$
    at shop $b_{\pc(i)}$ yields a solution for $\mathcal{I}'$, and it is simple to check that its cost is $n-k$.
    
        \noindent\fbox{$\Leftarrow$}
        Suppose that all books can be bought for a total cost of $n-k$.
    Since  $n$ books must be bought at unit price and shops only offer  a unit discount,
    $k$ shops must be chosen in the solution.
    Let $S' \subseteq S$  denote  these $k$ shops.
    Since $\mathscr{D}(s_i) = (1, d_G(u_i)+1)$ for each shop $s_i \in S$,
    we conclude that for each book $b_i \in B$ there is precisely one shop in
    $N[b_i] \cap S'$.
    Then 
    $\{u_i : s_i \in S'\}$ is a size-$k$ perfect code in $G$.
    
    Note that the number of visited shops corresponds exactly to the total discount received (i.e. to parameter $k$ in the reduction).
\qed
\end{proof}

    We now prove\footnote{See \Cref{app:proof-no-poly-kernel} for the proof.} the non-existence of polynomial kernels (under  standard complexity assumptions) for  \pblong parameterised by the   number of books. To this end, we use the {\sc or-composition}  technique~\cite{DBLP:conf/stacs/BodlaenderJK11}: given a problem $\mathscr P$ and a parameterised problem $\mathscr Q$, an {\sc or-composition} is a reduction taking $t$ instances $(I_1, \ldots, I_t)$ of $\mathscr P$, and building an instance $(J, k)$ of $\mathscr Q$, with $k$ bounded by a polynomial on $\max_{t'\leq t} |I_{t'}|+\log t$, such that $(J,k)$ is a yes-instance if and only if there exists $t'\leq t$ such that $I_{t'}$ is a yes-instance. 
    If 
    $\mathscr P$ is \NP-hard, then $\mathscr Q$ does not admit a polynomial kernel unless $\NP \subseteq \coNP/\mathrm{poly}$~\cite{DBLP:conf/stacs/BodlaenderJK11}.
    
	\begin{proposition}\label{prop:no-poly-kernel}
		\pblong admits no polynomial kernel unless $\NP \subseteq \coNP/\mathrm{poly}$.
	\end{proposition}
    
\section{Positive Results}

We now give exact algorithms for \pbabbrev: a polynomial-time algorithm for the case where every shop sells at most two books, and three parameterised algorithms 
based respectively on the number of books, the  number of shops, and a bound on the prices.

    We give a polynomial time algorithm for the case where each shop sells at most two books. As we shall see in \Cref{sec:approx}, this bound is best possible. Its running time is dominated by the time required to find a maximum matching in a graph with $|B\cup S|$ vertices.
    
    \begin{proposition}
	\label{proposition:Clever Shopper each store sells 2 books}
	\pbabbrev is in \P\ if every shop sells at most two books.
	\end{proposition}

\tikzstyle{svertex}=[circle,draw,fill=black!23,minimum size=17pt,inner sep=0pt,circular drop shadow]
\tikzstyle{bvertex}=[circle,draw,fill=black!3,minimum size=17pt,inner sep=0pt,circular drop shadow]
\tikzstyle{edge}=[draw,thick]
\tikzstyle{selected edge}=[draw, line width=2.5pt]

\begin{figure}[t]
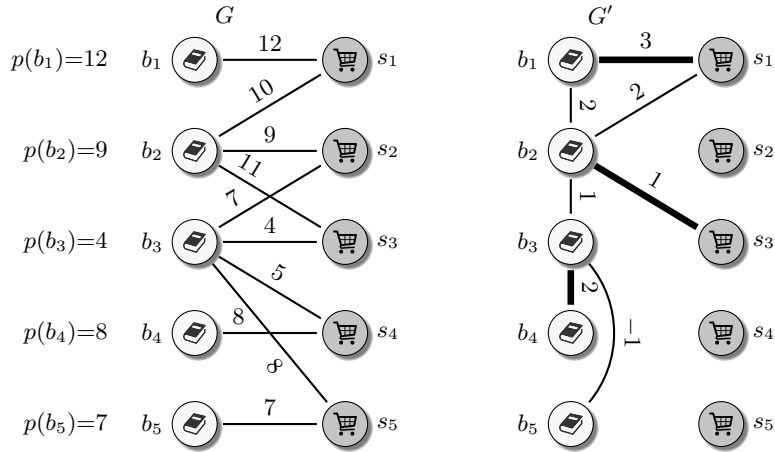

        \centering

    \begin{tikzpicture}
      [
        scale=0.2,
        shorten >=2pt,
        shorten <=2pt,
        node distance=.9in
      ]
            \begin{scope}
              \node at (2,3) {$G$};

        \node [bvertex] [label=left:$b_1$] (b1) at (0,0)
         {\includegraphics[width=.3cm]{book}};
        \node [bvertex,below = 6mm of b1] [label=left:$b_2$] (b2)
         {\includegraphics[width=.3cm]{book}};
        \node [bvertex,below = 6mm of b2] [label=left:$b_3$] (b3)
         {\includegraphics[width=.3cm]{book}};
        \node [bvertex,below = 6mm of b3] [label=left:$b_4$] (b4)
         {\includegraphics[width=.3cm]{book}};
        \node [bvertex,below = 6mm of b4] [label=left:$b_5$] (b5)
         {\includegraphics[width=.3cm]{book}};

        \foreach \p [count=\i] in {12,9,4,8,7}{
           \node [left= 3mm of b\i] {$p(b_\i){=}$\p\quad\ \ };
        }

        \node [svertex] [label=right:$s_1$] (s1) at (10,0)
          {\includegraphics[width=.3cm]{basket}};
        \node[svertex,below = 6mm of s1] [label=right:$s_2$] (s2)
          {\includegraphics[width=.3cm]{basket}};
        \node[svertex,below = 6mm of s2] [label=right:$s_3$] (s3)
          {\includegraphics[width=.3cm]{basket}};
                  \node[svertex,below = 6mm of s3] [label=right:$s_4$] (s4)
          {\includegraphics[width=.3cm]{basket}};
        \node[svertex,below = 6mm of s4] [label=right:$s_5$] (s5)
          {\includegraphics[width=.3cm]{basket}};

        \path [edge]
          (b1) edge [left] node [pos=0.5, sloped, above] {$12$} (s1);
        \path [edge]
          (b2) edge [right] node [pos=0.5, sloped, above] {$10$} (s1);
        \path [edge]
          (b2) edge [left] node [pos=0.5, sloped, above] {$9$} (s2);
        \path [edge]
                (b2) edge [right] node [pos=0.25, sloped, above] {$11$} (s3);
        \path [edge]
          (b3) edge [left] node [pos=0.25, sloped, above] {$7$} (s2);
        \path [edge]
                (b3) edge [right] node [pos=0.5, sloped, above] {$4$} (s3);
        \path [edge]
          (b3) edge [left] node [pos=0.5, sloped, above] {$5$} (s4);
        \path [edge]
          (b3) edge [left] node [pos=0.65, sloped, below] {$8$} (s5);
        \path [edge]
          (b4) edge [right] node [pos=0.20, sloped, above] {$8$} (s4);
        \path [edge]
          (b5) edge [left] node [pos=0.5, sloped, above] {$7$} (s5);
      \end{scope}

      \begin{scope}[xshift=25cm]
              \node at (2,3) {$G'$};
        \node [bvertex] [label=left:$b_1$] (b1) at (0,0)
         {\includegraphics[width=.3cm]{book}};
        \node [bvertex,below = 6mm of b1] [label=left:$b_2$] (b2)
         {\includegraphics[width=.3cm]{book}};
        \node [bvertex,below = 6mm of b2] [label=left:$b_3$] (b3)
         {\includegraphics[width=.3cm]{book}};
        \node [bvertex,below = 6mm of b3] [label=left:$b_4$] (b4)
         {\includegraphics[width=.3cm]{book}};
        \node [bvertex,below = 6mm of b4] [label=left:$b_5$] (b5)
         {\includegraphics[width=.3cm]{book}};

        \node [svertex] [label=right:$s_1$] (s1) at (10,0)
          {\includegraphics[width=.3cm]{basket}};
        \node[svertex,below = 6mm of s1] [label=right:$s_2$] (s2)
          {\includegraphics[width=.3cm]{basket}};
        \node[svertex,below = 6mm of s2] [label=right:$s_3$] (s3)
          {\includegraphics[width=.3cm]{basket}};
                  \node[svertex,below = 6mm of s3] [label=right:$s_4$] (s4)
          {\includegraphics[width=.3cm]{basket}};
        \node[svertex,below = 6mm of s4] [label=right:$s_5$] (s5)
          {\includegraphics[width=.3cm]{basket}};

        \path [selected edge]
          (b1) edge [left] node [pos=0.5, sloped, above] {$3$} (s1);
        \path [edge]
          (b2) edge [right] node [pos=0.5, sloped, above] {$2$} (s1);
        \path [selected edge]
                (b2) edge [right] node [pos=0.5, sloped, above] {$1$} (s3);
        \path [edge]
          (b1) edge [left] node [pos=0.5, sloped, above] {$2$} (b2);
        \path [edge]
          (b2) edge [left] node [pos=0.5, sloped, above] {$1$} (b3);
        \path [selected edge]
          (b3) edge [left] node [pos=0.5, sloped, above] {$2$} (b4);
        \path [edge]
          (b3) edge [left, bend left=40] node [pos=0.5, sloped, above] {$-1$} (b5);
      \end{scope}
    \end{tikzpicture}
  \caption{\label{fig:perfectmatching}%
    Each shop offers a discount of $3$  on
    a purchase of value $\ge 10$.
    Bold edges indicate how to obtain optimal discounts:
    buy book $b_1$ from shop $s_1$,
    book $b_2$ from shop $s_3$,
    and books $b_3$ and $b_4$ from shop $s_4$. The remaining books are bought at their cheapest available price (so here we buy $b_5$ from $s_5$).
    Our clever customer used the discounts to buy all books for $6$ less than if she had bought each  book at its lowest price: $3$ for $b_1$, $1$ for $b_2$, $2$ for $b_3$ and $b_4$ together.
  }
\end{figure}

    \begin{proof}
    Let $\mathcal{I}$ be an instance of \pbabbrev given by
    an edge-weighted bipartite graph $G=(B \cup S, E, w)$
    and a pair $(d_s, t_s)$ for each $s \in S$,
    where $d_s, t_s \in \mathbb{R^+}$. 
    Vertices in $S$ (resp. in $B$) have degree at most 2 (resp. at least 1). Note that vertices in $S$ can be made to have degree exactly 2, by adding dummy edges with arbitrarily high costs, with no impact on the solution. 
    For $b\in B$, let $p(b)$ be the cheapest available price for book $b$ (discount excluded), i.e., $p(b)=\min\{w(\{b, s\})\mid s\in S\}$.
      
  	Construct a new (non-bipartite) graph 
    $G' = (B \cup S, E', w')$, as follows:
    for every shop $s\in S$, let $\{b_1,b_2\}=N_G(s)$ (i.e., the two books available at shop $s$).    
   
    \begin{itemize}
    \item  For each $i\in \{1,2\}$, if $w(\{b_i,s\})\geq t_s$, then add an edge $\{b_i,s\}$ to $E'$ with weight $w'(\{b_i,s\})=d_s + p(b_i)-w(\{b_i,s\})$.  
    \item If  $w(\{b_1,s\})+w(\{b_2,s\})\geq t_s$, 
    add an edge $\{b_1,b_2\}$ to $E'$ with weight $w'(\{b_1,b_2\})=d_s+ p(b_1)-w(\{b_1,s\})+p(b_2)-w(\{b_2,s\})$.   If edge $\{b_1,b_2\}$ existed already, keep only the one with maximum weight. 
    \end{itemize}
    Note that edges with negative weights may remain: they may be safely ignored, but we keep them to avoid case distinctions in the rest of this proof. \Cref{fig:perfectmatching} illustrates the construction. 
Since a maximum weight matching for $G'$ can be found in polynomial time~\cite{Edmonds65}, it is now enough to prove the following claim: 
$G'$ admits a matching of weight at least $W$ if and only if instance $\mathcal{I}$ of \pbabbrev admits a solution of total cost at most $\sum_{b\in B} p(b)-W$.

        \noindent\fbox{$\Leftarrow$}
Assume that instance $\mathcal{I}$ admits a solution $E^*\subseteq E$ of total cost $\sum_{b\in B} p(b)-W$. Note that 
$W\geq 0$ 
(the sum of the minimum prices of the books is an upper bound of the optimal solution). We build a matching $M$ of $G'$ as follows. Let $s\in S$ be any \emph{discount shop}, i.e., a shop whose discount is claimed, and let $b_1$ and $b_2$ be its neighbours. Then at least one of them has to be bought from $s$ to get the discount. 
\begin{itemize} 
\item If $\{b_1,s\}\in E^*$ and $\{b_2,s\}\notin E^*$, add $\{b_1,s\}$ to $M$. The amount spent at this shop is $w(\{b_1,s\})-d_s = p(b_1)-w'(\{b_1,s\})$.
\item Similarly, if $\{b_2,s\}\in E^*$ and $\{b_1,s\}\notin E^*$, add $\{b_2,s\}$ to $M$. The amount spent at this shop is $w(\{b_2,s\})-d_s = p(b_2)-w'(\{b_2,s\})$.
\item Finally, if $\{b_1,s\}\in E^*$ and $\{b_2,s\}\in E^*$, then add $\{b_1,b_2\}$ to $M$. The amount spent at this shop is $w(\{b_1,s\})+w(\{b_2,s\})-d_s \geq p(b_1)+p(b_2)-w'(\{b_1,b_2\})$.
\end{itemize}
Note that edges added to $M$ are indeed present in $E'$, since in order to obtain the discount from $s$, the book prices must satisfy the same condition as for creating the corresponding edges. Note also that $M$ is a matching, since each book can be bought from at most one shop. 
Let $B^*$ be the set of books bought from discount shops. Summing over all these shops, the total  price paid for the books in $B^*$ is at least $\sum_{b\in B^*} p(b) - \sum_{e\in M}w'(e)$. 

The books in $B\setminus B^*$ do not yield any discount, so the total price paid for them is at least $\sum_{b\in B\setminus B^*} p(b)$. Overall, the cost of the books is at least $\sum_{b\in B} p_b- \sum_{e\in M}w'(e)$, therefore $\sum_{e\in M}w'(e)\geq W$.

\noindent\fbox{$\Rightarrow$}
   Let $M$ be a maximum weight matching of $G'$ of weight $W$. For each edge $e\in M$, let $s_e$ be the shop for which $e$ was introduced. For an edge $e=\{b,s_e\}\in M$, buy book $b$ from shop $s_e$. The price is sufficient to reach the threshold for the discount, so we pay $w(\{b,s_e\})-d_e = p(b)-w'(e)$. For an edge $e=\{b_1,b_2\}\in M$, buy books $b_1$ and $b_2$ together from shop $s_e$. We again get the discount, and pay $w(\{b_1,s_e\})+w(\{b_2,s_e\})-d_e = p(b_1)+p(b_2)-w'(e)$. 
   Note that for $e\neq f\in M$, $s_e\neq s_f$, so we never count the same discount twice. For every other book, buy them at the cheapest possible price $p(b)$, without expecting to get any discount.  The total price paid is at most $\sum_{b\in B}  p(b)- \sum_{e\in M}w'(e)= \sum_{b\in B}  p(b) -W$.  
 \qed
 \end {proof}
	
	We now give a dynamic programming \FPT\ algorithm with the number of books as parameter.
	\begin{proposition}
		\pblong admits an \FPT\ algorithm for parameter $n$ with running time $O(m 3^n)$.
	\end{proposition}
	
	\begin{proof}
		Given $j\in[m]$ and $B'\subseteq B$, let $p_j(B')$ be the price for buying all books in $B'$ together from shop $s_j$ (discount included),
		and
		$p_{ \leq j}(B')$  be the lowest price that can be obtained when purchasing all books in $B'$ 
        from a subset of $\{s_1, \ldots, s_j\}$. 
		Our goal is to compute $p_{\leq m}(B)$.
		
		For $j=1$, clearly $p_{ \leq 1}(B')=p_1(B')$ for every $B'$. For any other $j$, consider an optimal way of buying the books in $B'$ from shops $s_1,\ldots, s_j$. 
		This way the customer buys some (possibly empty) subset $B''$ of books in $s_j$, and the rest, i.e., $B'\setminus B''$, at the lowest price from shops $s_1,\ldots, s_{j-1}$. 
        Therefore: 
		$$p_{ \leq j}(B') = \left\{\begin{array}{ll}
		p_j(B')& \mathrm{ if}\ j=1, \\
		\min_{B''\subseteq B'}\{p_j(B'')+p_{\leq j-1}(B'\setminus B'')\}& \mathrm{ otherwise}.
		\end{array}
		\right.
		$$
		
		The values of $p_j(B')$ for all $j$ and $B'$ can be computed in $O(m2^n)$ time. Then the dynamic programming table requires to enumerate, for all $j$, all subsets $B'$ and $B''$ such that $B''\subseteq B' \subseteq B$. 
		Any such pair $B'', B'$ can be interpreted as a vector $v\in\{0,1,2\}^n$, where $i\in B''\Leftrightarrow v_i=2$ and $i\in B'\Leftrightarrow v_i\geq 1$.
		Therefore, filling the dynamic table takes $m3^n$ steps, each requiring constant time.
	\qed
    \end{proof}
	
	As usual with dynamic programming, this algorithm yields the optimal price that can be obtained. One gets the actual solution (i.e., where to buy each book) with classic backtracing techniques.
	
The \NP-hardness of \pblong for two shops (using large prices, encoded in binary) and its \W[1]-hardness when the parameter is the number of shops leave a very small opening for positive results: we can only consider small prices (encoded in unary) for a constant number of shops. The following result proves the tractability of this case.

\begin{proposition}	\label{prop:xp-fewshops}		
\pblong  admits an \XP\ algorithm running in time $O(nm\mathcal W^m)$, where $\mathcal W$ 
is the sum of all the prices of the instance, $n$ is the number of books, and $m$ is the number of shops.
\end{proposition}
\begin{proof}
 We propose the following dynamic programming algorithm, which generalises the classical pseudo-polynomial algorithm for {\sc Partition}.  
 Let $i\in [n]$ and $p_s\in[\mathcal W]$ for $s\in S$. Define  $T[i,p_{s_1}, \ldots, p_{s_m}]$ as 1 if it is possible to buy books $1$ to $i$ by paying exactly $p_s$ (discount excluded) in shop $s$; and 0 otherwise. 
 For $i=0$, $T[0,p_{s_1}, \ldots, p_{s_m}]=1$ if and only if $p_s=0$ for all $s\in S$. The following formula allows to fill the table recursively for $i\geq 1$:
 \[T[i,p_{s_1}, \ldots, p_{s_m}]=\max_{e\in E, i\in e} T[i-1, p_{s_1}', \ldots, p_{s_m}'] \text{ where } p_s'= 
 \left\{\begin{array}{ll}
		 p_s-w(e)  &\text{ if } s\in e, \\
		p_s  &\text{ otherwise.}
		\end{array}
		\right. \]

It remains to check whether there exists a valid solution within the table. To this end, we need to take the discounts into account. Clearly, an entry $T[n,p_{s_1}, \ldots, p_{s_m}]=1$ leads to a solution if the following holds:
$$\sum_{s\in S}p_s - \sum_{s\in S, p_s\geq t_s}d_s\leq K.$$

The running time corresponds exactly to the time needed to fill the table:  any of the $n\mathcal W^m$ cells requires at most $m$ look-ups, which yields the claimed running time.
\qed
\end{proof}

\begin{proposition}	\label{prop:fpt-fewshops}	
\pblong  admits an \FPT algorithm for parameter $m$ when all prices are equal.
\end{proposition}
\begin{proof}
We assume without loss of generality that all prices are equal to 1. Let $S'\subseteq S$. We write $f_{S'}:B\cup S\rightarrow \mathbb N$ for the following function:
\begin{align*}
f_{S'}(b)&=1\text{ for }b\in B,\\
f_{S'}(s)&=t_s\text{ for }s\in S',\\
f_{S'}(s)&=0\text{ for }s\notin S'.
\end{align*}
We write $d_{S'}=\sum_{s\in S'} d_s$ and $t_{S'}=\sum_{s\in S'} t_s$. An \emph{$f$-star subgraph} of $G=(B\cup S, E)$ is a subgraph $G'$ such that the degree of each vertex $u\in B \cup S$ is at most $f(u)$ in $G'$, and every connected component of $G'$ is isomorphic to $K_{1,p}$ for some integer $p$.

Let $\mathcal{I}=(B\cup S, E, w, \discounts, K)$ be an instance  of \pblong  with $w(e)=1$ for all $e\in E$. We show that $\mathcal{I}$ is a yes-instance if and only if there exists $S'\subseteq S$ with $|B| - d_{S'} \leq K$ such that $(B \cup S, E)$ admits an $f_{S'}$-star subgraph with $t_{S'}$ edges. An \FPT\ algorithm follows easily from this characterisation: enumerate all subsets $S'$ of $S$ in time $2^{|S|}$, and for each subset,  compute a maximum $f_{S'}$-star subgraph in time $O(|E| \log |B\cup S|)$~\cite{DBLP:journals/ipl/Gabow76a}.

        \noindent\fbox{$\Rightarrow$}
Let $E'\subseteq E$ be a solution and $S'$ be the set of shops whose threshold $t_s$ is reached. 
Since the total price  is $|B|-d_{S'}$, we have $|B| - d_{S'} \leq K$.
Since every weight  equals 1, all vertices of $S'$ have degree at most $t_s$ in $E'$. Let $E''\subseteq E'$ be a subset obtained by keeping exactly $t_s$ edges incident to each $s\in S'$ and no edge incident to $s\notin S$. Then $E''$ is an $f_{S'}$-star subgraph of size $t_{S'}$.

        \noindent\fbox{$\Leftarrow$}
Let $G'=(B\cup S, E')$ be an $f_{S'}$-star factor of $G$ of size $t_{S'}$ with $S'\subseteq S$, and $|B| - d_{S'} \leq K$. The degree and size constraints force all vertices 
in $S'$ 
to have degree exactly $t_s$ in $G'$. We build a solution as follows: for each book $b\in B$, if $E'$ contains an edge $(b,s)$ incident to $b$, then buy $b$ from shop $s$, otherwise buy $b$ from any other shop. Overall, at least $t_s$ books are purchased from a shop $s\in S'$, so the total price is at most $|B|-d_{S'}$.
\qed
\end{proof}
\section{Approximations}\label{sec:approx}
Since variants of
{\sc Clever Shopper} are, by and large, hard to solve exactly, it is natural to look for approximation algorithms. However, 
our hardness proofs can be modified to imply 
the \NP-hardness of deciding whether the total price (including discounts) is 0 or more. For instance, in \Cref{prop:problem-is-hard-for-two-shops}, we can set the discounts to $T/2$ instead of $1$, so the {\sc Partition} instance reduces to checking whether the optimal solution has cost  0. Therefore, we start with the following bad news:

\begin{corollary}
{\sc Clever Shopper} admits no approximation  unless \P\ = \NP. 
\end{corollary}

Since this result seems resilient to most natural restrictions on the input structure (bounded prices, bounded degree, etc.), our proposed angle is 
to maximise the total discount rather than minimise the total cost. However, maximising the total discount is only relevant when the base price of the books is the same in all solutions (otherwise the optimal solution might not be the one with maximum discount), i.e., each book $b$ has a fixed price $w_b$, and $w(\{b,s\})=w_b$ for every $\{b,s\}\in E$. We call this variant {\sc Max-Discount Clever Shopper}. This ``fixed price'' constraint is not strong (all reductions from \Cref{sec:hardness-results} satisfy it). In this setting, \Cref{prop:problem-is-hard-for-two-shops} shows that it is \NP-hard to decide whether the optimal discount is 1 or 2. This yields the following corollary:

\begin{corollary}
{\sc Max-Discount Clever Shopper} is \APX-hard: it does not admit a $(2-\epsilon)$-approximation unless \P\ = \NP.
\end{corollary}

Whether or not {\sc Max-Discount Clever Shopper} admits a fixed-ratio approximation remains open.

 \begin{proposition}\label{prop:apx-hardness}
   {\sc Max-Discount Clever Shopper} is \APX-hard even when each shop sells at most $3$ books, 
   and each book is available in at most 2 shops.
 \end{proposition}
 
 \begin{proof}
 We reduce from {\sc Max $3$-Sat} (the problem of satisfying the maximum number of clauses in a {\sc 3-sat} instance), known to be \APX-hard 
 when each literal occurs 
   exactly twice~\cite{DBLP:journals/eccc/ECCC-TR03-049}.  
   Let $\varphi = C_1 \wedge C_2 \wedge \cdots \wedge C_m$ be such a $3$-CNF formula over a set $X = \{x_1, x_2, \ldots, x_n\}$ of boolean variables.
   For every $1 \leq i \leq m$ and $1 \leq j \leq 3$,  let $\ell_{i,j}$ be the $j$-th literal of clause $C_i$.
   We obtain an instance $\mathcal{I}$ of {\sc Max-Discount Clever Shopper} by first building a bipartite graph $G = (B\cup S, E)$ as follows
   (for ease of presentation, $C_i$, $x_i$ and $\ell_{i,j}$ will be used both to denote respectively clauses, variables and literals in 
   $3$-CNF formula context, and the corresponding vertices in $G$):
   \begin{align*}
      B &= \{\ell_{i,j} : 1 \leq i \leq m \text{ and } 1 \leq j \leq 3\} \;\cup\; \{x_i : 1 \leq i \leq n\} \\
      S &= \{C_i : 1 \leq i \leq m\} \;\cup\; \{t_i, f_i : 1 \leq i \leq n\} \\
      E &= E_1 \;\cup\; E_{2,p} \;\cup\; E_{2,n} \;\cup\; E_3 \\
      \intertext{where}
      E_1 &= \{\{\ell_{i,j}, C_i\} : 1 \leq i \leq m \text{ and } 1 \leq j \leq 3\} \\
      E_{2,p} &= \{\{\ell_{i,j}, t_i\} : 1 \leq i \leq m \text{ and $\ell_{i,j}$ is the positive literal $x_i$}\} \\
      E_{2,n} &= \{\{\ell_{i,j}, f_i\} : 1 \leq i \leq m \text{ and $\ell_{i,j}$ is the negative literal $\overline{x_i}$}\} \\
      E_3 &= \{\{x_i, t_i\}, \{x_i, f_i\} : 1 \leq i \leq n\}\text{.}
   \end{align*}
   Observe that each shop sells exactly $3$ books and that each book is sold in exactly $2$ shops.
   We now turn to defining the prices, the thresholds and the discounts.
   All shops sell books at a unit price.
   For the shops $C_i$, $1 \leq i \leq m$, a purchase of value $1$ yields a discount of $1$.
   For the shops $t_i$ and $f_i$, $1\leq i \leq n$, a purchase of value $3$ yields a discount of $2$.
   This discount policy implies that, for every $1 \leq i \leq n$, a customer cannot obtain 
   a $2$ discount both in shop $t_i$ and in shop $f_i$ (this follows from the fact that the book $x_i$ is sold  by both shops 
   $t_i$ and $f_i$).

  First, it is easy to see that 
  the largest  discount that can be obtained is $2n + m$  
  (the upper bound is achieved by obtaining a discount in every shop $C_i$ for $1 \leq i \leq m$, and in either the shop $t_i$ 
  or the shop $f_i$ for $1 \leq i \leq n$).
  On the other side, for any truth assignment $\tau$ for $\varphi$ satisfying $k$ clauses, a $2n + k$ discount can be obtained as follows.
  \begin{itemize}
      \item
      For any variable $x_i$, $1 \leq i \leq n$,
      if $\tau(x_i) = \texttt{false}$, then buy $3$ books from shop $t_i$, and
      if $\tau(x_i) = \texttt{true}$ then  buy $3$ books from shop $f_i$. Intuitively, if a variable is true, then all negative literals are ``removed'' by $f_i$, and all positive literals remain available for the corresponding clauses.
      \item
      For any clause $C_i = \ell_{i,1} \vee \ell_{i,2} \vee \ell_{i,3}$ satisfied by the truth
      assignment $\tau$, buy book $\ell_{i,j}$ from shop $C_i$, where $\ell_{i,j}$ is a
      literal satisfying the clause $C_i$.
  \end{itemize}
  Then it follows that 
  \begin{align*}
      \OPT(I) 
      = 2n + \OPT(\varphi)  &= 3m/2 + \OPT(\varphi)      &\qquad& \text{(since $4n=3m$)}\\
      &\leq 3 \OPT(\varphi) + \OPT(\varphi) &\qquad& \text{(since $2 \OPT(\varphi) \geq m$)} \\
      &\leq 4 \OPT(\varphi)\text{.} &&
  \end{align*}   

  Suppose now that we buy all books in $B$ for a total  discount of $k'$.
  First, we may clearly assume that $k' \geq 2n$ since a total $2n$ discount can always be achieved by buying  $3$ books  either from shop $t_i$ or from shop $f_i$, for every $1 \leq i \leq n$.
  Second, we may also assume that, for every $1 \leq i \leq n$, we buy either exactly $3$ books from shop 
  $t_i$ or exactly $3$ books  from shop  $f_i$.
  Indeed, if there exists an index $1 \leq i \leq n$ for which this is false, 
  then buying either exactly $3$ books from shop 
  $t_i$ or exactly $3$ books from shop $f_i$ instead results in a total $k''$ discount with $k'' \geq k'$
  (this follows from the fact that we can get a $2$ discount from $t_i$ or $f_i$ but only a $1$ discount from
  any shop $C_j$, $1 \leq j \leq m$). 
  We now obtain a truth assignment $\tau$ for $\varphi$ as follows:
  for any variable $x_i$, $1 \leq i \leq n$,
  set $\tau(x_i) = \texttt{false}$ if we buy $3$ books from shop $t_i$, and
  set $\tau(x_i) = \texttt{true}$ if we buy $3$ books from shop $f_i$ 
  (the truth assignment $\tau$ is well-defined since, for $1 \leq i \leq n$, we  cannot simultaneously buy $3$ books
  from shop $t_i$ and $3$ books from shop $f_i$ because of book $x_i$). 
  Therefore, a clause $C_i$ is satisfied by $\tau$ if and only if the corresponding shop $C_i$ contains at least one book $l_{i,j}$ which is not bought from some other shop $t_i$ or $f_i$. 
  If we let $k$ stand for the number of clauses satisfied by $\tau$, then we obtain $k \geq k' - 2n$.
  It then follows that
\[  
      \OPT(\varphi) - k 
      = \OPT(I) - 2n - k 
      \leq \OPT(I) - 2n - k' + 2n 
      = \OPT(I) - k'\text{.} 
\]
Therefore, our reduction  is an \L-reduction 
  (\emph{i.e.}, $\OPT(\mathcal{I}) \leq \alpha_1 \OPT(\varphi)$ and
  $\OPT(\varphi) - k \leq \alpha_2 \left(\OPT(\mathcal{I}) - k'\right)$)
  with $\alpha_1 = 4$ and $\alpha_2 = 1$.
  \qed
\end{proof}

 \begin{proposition}
   {\sc Max-Discount Clever Shopper} where each shop sells at most $k$ books admits a $k$-approximation.
 \end{proposition}
\begin{proof}
Let $B_s$ be the set of books sold by shop $s$. 
Our approximation algorithm proceeds as follows: start with a set of selected shops $S'=\emptyset$,  a set of available books $B'=B$ and sort the shops by decreasing value of $d_s$. Then for each shop $s$, let $B_s'=B_s\cap B'$. If the books in $B_s'$ are enough to get the discount ($\sum_{b\in B'_s}\geq t_s$), then assign all books of $B_s'$ to shop $s$, add $s$ to $S'$ and set $B'=B'\setminus B'_s$. Finally, assign the remaining books to arbitrary shops that sell them.
 
We now prove the approximation ratio. For any $b\in B$, if $b\in B'_s$ for some $s\in S'$ then let $\delta(b) = d_s$, and $\delta(b)=0$ otherwise. Thus, for any shop $s\in S'$, $d_s=\frac{1}{|B'_s|}\sum_{b\in B'_s}\delta(b) \geq \frac{1}{k}\sum_{b\in B'_s}\delta(b)$ due to the degree-$k$ constraint. Note that for each shop of $S'$, the amount spent at  $s$ is at least $t_s$, so the total discount obtained with this algorithm is $D\geq \sum_{s\in S'} d_s\geq \frac 1k \sum_{b\in B} \delta(b)$ 

We now compare the result of the algorithm with any optimal solution. For such a solution, let $D^*$ be its total discount, $S^*$ be the set of shops where purchases reach the threshold, and, for any $s\in S^*$, let $B^*_s$ be the (non-empty) set of books purchased in shop $s$. Note that $D^*=\sum_{s\in S^*} d_s$. 
 
Consider a shop $s\in S^*$. We show that there exists a book  $b^*(s)\in B^*_s$ with $\delta(b^*(s))\geq d_s$. 
If 
$s\in S^*\cap S'$, then we take $b^*(s)$ to be any book in $B^*_{s}$. Either $b^*(s)\in B'_s$, in which case  $\delta(b^*(s)) = d_s$, or $b^*(s)\notin B'_s$, in which case $b^*(s)$ was assigned by the algorithm to a shop with a larger discount, i.e., $\delta(b^*(s)) \geq d_s$. 
If $s\in S^*\setminus S'$, since $s\notin S'$, 
at least one book in $B^*_s$ 
is not available at the time the algorithm considers shop $s$; let $b^*(s)$ be such a book. Since it is not available, it has been selected as part of $B'_{s'}$ for some earlier shop $s'$ (i.e.,  $d_s\leq d_{s'}$). Therefore, $b^*(s)\in B^*_{s}\cap B'_{s'}$ and $\delta(b^*(s)) = d_{s'}\geq d_s$. 
Since the sets $B^*_s$ are pairwise disjoint for $s\in S^*$, we have $\sum_{s\in S^*} \delta(b^*(s)) \leq \sum_{b\in B} \delta(b)$. Putting it all together, we obtain:
\begin{equation*}
D^*= \sum_{s\in S^*} d_s 
\leq \sum_{s\in S^*} \delta(b^*(s)) 
\leq \sum_{b\in B} \delta(b)
\leq k D\text{.}
\end{equation*}
\qed
\end{proof}

\section{Conclusion}

We introduced 
the \pbabbrev problem, a variant of {\sc Internet Shopping} with free deliveries and shop-specific discounts 
based 
on shop-specific thresholds. We proved a number of hardness results, both in the classical complexity setting and from a parameterised complexity point of view. We also gave efficient algorithms for particular cases where restrictions apply to the number of books, the number of shops, or the nature of prices.

An interesting angle for future work is that of designing efficient exact algorithms for the general cases in which our \FPT\ algorithms are not sufficient.
Furthermore, it would be of interest to determine whether 
the \pbabbrev problem is \FPT\ for parameter 
\emph{maximum price $+$ number of shops}.

	\bibliographystyle{mynatstyle}
	\bibliography{main.bib}
	\newpage
    \appendix
    \section{Proof of \Cref{prop:no-poly-kernel}}\label{app:proof-no-poly-kernel}

    The following classical \NP-complete problem~\cite{DBLP:conf/coco/Karp72}  will be useful in that regard. 
	
        \medskip

	\decisionproblem{\sc Exact Cover By $3$-Sets (X3C)}{ 
		a set $X = \{x_1, x_2, \ldots, x_{3m}\}$ of items, a collection $\mathscr{C}$ of $3$-sets of $X$.}{is there a subset $\mathscr{C}'$ of 
		$\mathscr{C}$ that covers
		each item of $X$ exactly once?}
	{\sc X3C} remains  
    	\NP-complete  
	 when each 
     $x_i$ appears in exactly $3$ sets of
	$\mathscr{C}$~\cite{DBLP:journals/tcs/Gonzalez85}.

    	\begin{proof}We build an {\sc or-composition} 
        using {\sc Exact Cover By 3-Sets}.
		Consider $t$ instances of {\sc X3C}
         over the same 
         number~$n$ of items. 
        They are represented as bipartite graphs $(S_h\cup [n], E_h)$ for $h\in[t]$, 
        where the 3-sets of $[n]$ are represented as degree-3 vertices $u\in S_h$.
		
		We first define some ``shop identifiers''. Write
		$J=\{0,1\}\times [\lceil\log t\rceil]$. For each integer $h\in [t]$, 
		let $\mathrm{Key}_h$ be the size-$\lceil\log t\rceil$ subset of $J$
		containing $(b, j)$ if the $j$th digit in the binary representation
		of $h$ is equal to $b$. Note that for $1\leq h< h'  \leq t$, we have
		$\mathrm{Key}_h\neq \mathrm{Key}_{h'}$. 
		We now build a new instance $(B\cup S, E, w,\discounts, K)$ as follows:
		\begin{itemize}
			\item Create shops $\sigma_{j}$ for all $j\in J$. Let
			$\Sigma= \{\sigma_{j}\mid j\in J\}$. The global set of shops is
			$S = \Sigma \cup \bigcup_{h\in[t]} S_h$. Note that $|S|=t+2\log t$.
			\item Create books $x^i_{j}$ for all $i\in [n]$ and $j\in J$. The
			global set of
			books is
			$B=\{x^i_j\mid i\in [n], j\in J\} \cup [n]$.
			Let $n'=|B|=n(2\log t +1)$.
			\item For each edge $e=\{s,i\}\in E_h$, where $h\in [t], s\in S_h, i\in[n]$,
			add edges $\{s,i\}$ and  $\{s, x^i_{ j}\}$ for all $j\in \mathrm{Key}_h$.
			Add also edges $\{\sigma_j,x^i_{j}\}$ for all $i\in[n]$ and $j\in J$.
			The overall set of edges is denoted $E$.
			\item Let all costs be equal to $T+1$ where $T$ is any non-negative
			integer (say, $T=42$).
			\item Let shop $s\in S$ have $t_s=k(T+1)$ and $d_s=k$ where $k$ is the degree of $s$. In other words, all shops give a discount of 1 per book only the buyer buys all books available from the shop.
			\item The budget is $K=Tn'$.
		\end{itemize}
		
		Note that due to the pricing and discount functions, the average cost
		of a book in a shop is between $T$ and $T+1$, and it reaches $T$ if and only
		if all books in this shop have been purchased. Since the shopper needs to
		buy $n'$ books with a budget of $Tn'$, she must either buy all books
		from the shops she visits, or none at all. Thus she is faced with the problem of
		finding a set of shops whose available books correspond exactly to
		the set of books she needs.
		
		The intuitive idea behind our construction is the following. We first
		build a set of shops behaving exactly like the union of all the sets in
		the instances of {\sc X3C}. This is achieved directly by selling
		the original books in the corresponding shops: the fact that each book
		must be taken exactly once directly gives an exact cover. The rest of the
		shops ($\Sigma$) and books ($x^i_j$) ensure that shops are used from a single
		set $S_h$. More precisely, half of the books $x^i_j$ are purchased together
		with book $i$, and they correspond to an identifier of the shops, while the
		other half must be purchased in shops from $\Sigma$, and enforce that the
		identifiers are the same for all books. That is, all books are purchased
		in shops from the same set $S_h$, which yields a solution to the $h$th instance
		of {\sc X3C}.
		
		We now formally prove that there is a solution to this instance of
		\pblong if and only if \emph{some} instance
		$(S_h\cup [n], E_h)$ of  {\sc Exact Cover By 3-Sets} for $h\in[t]$,
		is a yes-instance, which completes the {\sc or-composition}. 
		
        \noindent\fbox{$\Leftarrow$}\;  
        Let $h\in [t]$ be such that $(S_h\cup [n], E_h)$ is a
			yes-instance. Let $S_h'$ be the solution (that is, a subset of $S_h$
			such that all vertices in $[n]$ have exactly one neighbour in $S_h'$).
			Let $\Sigma'=\{\sigma_j\mid j\in J \setminus \mathrm{Key}_h\}$. We
			show that buying all books in all shops of $S' = S'_h\cup \Sigma'$
			gives a valid solution.
			Pick $i\in [n]$.  Since
			$S_h'$ is a solution to the {\sc X3C} instance, then there
			exists a single $s\in S_h'$ such that $(s,i)\in E$.
			Books $i$ and $x^i_{j}$ for
			$j\in \mathrm{Key}_h$ are thus sold by shop $s$, but not by any other shop in
			$S'$ (in particular, not by any shop in $\Sigma'$ since $j\notin \mathrm{Key}_h$).
			Consider now a book
			$x^i_{j}$ with $i\in [n]$ and $j\in J\setminus \mathrm{Key}_h$.
			Then $x^i_{j}$ is sold by shop $\sigma_j\in \Sigma'\subset S'$, and by no other
			shop in $S'$. Overall, each book is sold by exactly one shop in
			$S'$, so the shopper buys all books from those shops, for a base
			price of $n'(T+1)$ and with a discount of $n'$.
			
        \noindent\fbox{$\Rightarrow$}\;  
            As it has been remarked already, in any solution, the
			set of shops $S'\subseteq S$ must contain each book exactly once.
			Consider first book $1$: it is sold by a shop $s_1\in S'\cap S_{h_1}$
			for some $h_1\in [t]$ (since no shop in $\Sigma$ sells books in
			$[n]$). Consider now books $x^1_{j}$, with $j\in J$. If
			$j\in \mathrm{Key}_{h_1}$, then $x^1_{j}$ is sold by $s_1$,
			and thus $\sigma_j\notin S'$. If $j\in J\setminus\mathrm{Key}_{h_1}$,
			then $x^1_{j}$ is not sold by $s_1$, nor by any other shop  in
			$S'\setminus \Sigma$ (such a shop would also sell book $1$, which
			is already taken from shop $s_1$). Thus $x^1_{j}$ must be sold by some
			shop from $\Sigma$, which can only be $\sigma_j$ by construction. Hence
			$\{\sigma_j\mid J\setminus\mathrm{Key}_h\} \subseteq S'$. Consider now
			any index $h_2\neq {h_1}$. There exists some
			$j\in\mathrm{Key}_{h_2}\setminus\mathrm{Key}_{h_1}$. If there exists
			some shop $s\in S_{h_2}\cap S'$, then $s$ sells book $x^i_{ j}$ for
			some $i\in [n]$. However, this book is already taken at shop
			$\sigma_j\in S'$ (since $j\in J\setminus \mathrm{Key}_{h_1}$), hence
			there is no such shop $s$. Overall, $S'\setminus \Sigma \subseteq S_{h_1}$,
			that is, the shopper uses only shops from the same set $S_{h_1}$ as well
			as some shops from $\Sigma'$.
			We can now prove that $S'_{h_1}=S'\cap S_{h_1} = S'\setminus \Sigma$
			is a solution to the instance $(S_{h_1}\cup [n], E_{h_1})$ of
			{\sc Exact Cover By 3-Sets}. Indeed, consider any $i\in [n]$, then
			it is sold by a single shop in $s\in S'$, which cannot be in
			$\Sigma$, hence $s\in S'_{h_1}$. In other words, there is exactly
			one $s\in S'_{h_1}$ such that $(s, i)\in E_{h_1}$, so  $S'_{h_1}$
			is a valid cover of $[n]$.
    \qed        
	\end{proof}
	\newcommand\equals{=}
	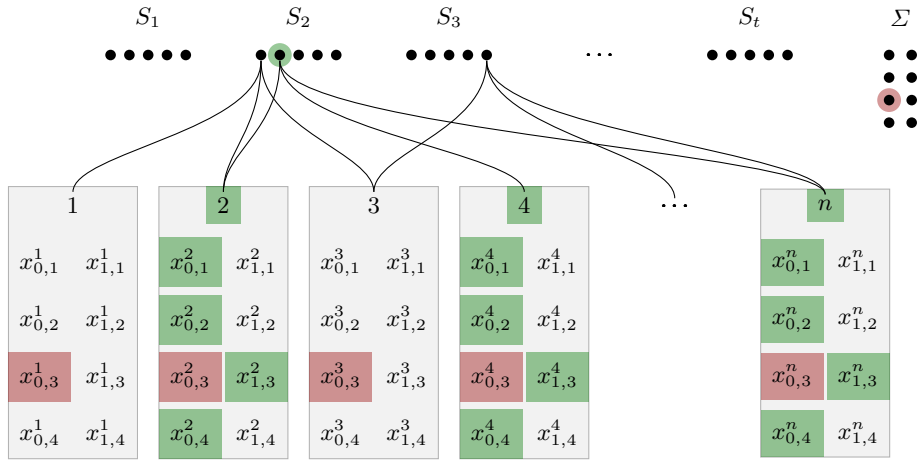
\begin{figure}\centering
		\begin{tikzpicture}[inner sep=2pt]
		\foreach \h in {1, ...,4} {
			\begin{scope}[shift={(\ifnum\h>3 2*\h+2\else 2*\h\fi+1,0)}]
			\def\label{\ifnum\h\equals4 $S_t$\else \ifnum\h\equals5 $\Sigma$\else $S_\h$ \fi\fi}
			\node at (0, 2.5) {\label};
			\foreach \s in {1, ..., 5} {
				\node (s\h\s) at (\s*0.25-0.75,2) {};
				\fill (s\h\s) circle [radius=2pt];
			}
			\end{scope}
		}
		\def\h{5}
		\begin{scope}[shift={(13,0)}]
		\def\label{$\Sigma$}
		\node at (0, 2.5) {\label};
		\foreach \s in {1, ..., 4} {
			\node (s\h0\s) at (-0.15,2.3-\s*0.3) {};
			\node (s\h1\s) at (0.15,2.3-\s*0.3) {};
			\fill (s\h0\s) circle [radius=2pt];
			\fill (s\h1\s) circle [radius=2pt];
		}
		\end{scope}

		\foreach \i in {1, ..., 5} {
			\begin{scope}[shift={(\ifnum\i\equals5 12\else 2*\i\fi,0)}]
			\def\label{\ifnum\i\equals5 n\else \i\fi}
			\node (i\i) at (0,0){$\label$};
			\begin{scope}[shift={(0,0)}]
			\foreach \j in {1, ..., 4} {
				\node (x\i0\j) at (-0.44,-\j*0.76){$x^{\label}_{0,\j}$};
				\node (x\i1\j) at (+0.44,-\j*0.76){$x^{\label}_{1,\j}$};
			}
			\end{scope}
			\end{scope}
			\node (ix) at (10,0) {$\ldots$};
			\node (sx) at (9,2) {$\ldots$};
		}
		
		\foreach \s/\i in {21/1,21/2,21/3, 22/2,22/4,22/5,35/3,35/x,35/5} {
			\draw (s\s) .. controls ($ (s\s) +(0,-1)$) and ($ (i\i) +(0,0.7)$) .. (i\i);
		}
		\definecolor{darkgreen}{RGB}{0,120,0}
		\definecolor{darkred}{RGB}{150,0,0}
		\begin{pgfonlayer}{bg}    
		\foreach \i in {1, ..., 5} {
			\node[draw=black!30,fill=black!5, fit=(i\i) (x\i04) (x\i14)] {};
		}
		\foreach \n in {i2,i4,i5,x201,x202,x213,x204,x401,x402,x413,x404,x501,x502,x513,x504} {
			\node[fill=darkgreen, opacity=0.4, fit=(\n)] {} ;
		}
		\fill[darkgreen, opacity=0.4]  (s22) circle[radius=4.5pt];
		\foreach \n in {x103,x203,x303,x403,x503} {
			\node[fill=darkred, opacity=0.4, fit=(\n)] {} ;
		}
		\fill[darkred, opacity=0.4]  (s503) circle[radius=4.5pt];
		\end{pgfonlayer}
		
		\end{tikzpicture}
		\caption{Illustration of the reduction from {\sc Exact Cover By 3-sets}.
			The instances are drawn in the top part, as bipartite graphs between $S_h$
			and books $[n]$ (for better readability, most edges are ommited, but in
			fact all vertices in sets $S_h$ have degree 3). The reduction introduces
			shops in $\Sigma$, and books $x^i_j$. Books sold in one of the shops
			from $S_2$ are highlighted in green. They correspond to the books $i$
			to which this element is connected in the corresponding instance of
			{\sc Exact Cover}, as well as books $x^i_j$ where $j$ visits the elements
			of $\mathrm {Key}_2$ (the ``identifier'' of $S_2$). Since $2$ is written
			0010 in binary, the key contains positions $(0,1),(0,2),(1,3),(0,4)$.
			Books sold by shop $\sigma_{0,3}\in\Sigma$ are highlighted in red. They are books
			$x^i_{0,3}$ for all $i$. Notice that those two shops have no books in
			common, since $(0,3)\notin \mathrm {Key}_2$.
		}
	\end{figure}

\end{document}